\documentclass[
final
]{dmtcs-episciences}

\usepackage[utf8]{inputenc}
\usepackage{subfigure}
\newcommand{\comment}[1]{}
\newcommand{\N}{\mathbb{N}}

\newcommand{\harp}{\upharpoonright}
\newcommand{\deepc}{\text{-deep}_C}
\newcommand{\deepk}{\text{-deep}_K}
\usepackage {amsmath,amssymb,latexsym}
\usepackage {delarray,graphics,graphpap,alltt,graphicx}
\def\varphi{\Phi}
\def\squareforqed{\hbox{\rlap{$\sqcap$}$\sqcup$}}
\def\qed{\ifmmode\squareforqed\else{\unskip\nobreak\hfil\penalty50\hskip1em\null\nobreak\hfil\squareforqed\parfillskip=0pt\finalhyphendemerits=0\endgraf}\fi}

\newenvironment{proof}{\noindent\textbf{Proof.}}{}

\newtheorem{theorem}{Theorem}[section]

\newtheorem{atheorem}{Theorem}{\bf }{\it }
\newtheorem{adefinition}[atheorem]{Definition}{\bf }{\rm }
\newtheorem{alemma}[atheorem]{Lemma}{\bf }{\it }
\newtheorem{acorollary}[atheorem]{Corollary}{\bf }{\it }

\author{Philippe Moser \affiliationmark{1}\thanks{P.~Moser was on Sabbatical
Leave to the National University of Singapore, supported in part by SFI
Stokes Professorship and Lectureship Programme.}
  \and Frank Stephan \affiliationmark{2}\thanks{F.~Stephan was supported
in part by NUS grants R146-000-181-112 and R146-000-184-112.}
  }
\title[Depth, Highness and DNR degrees]{Depth, Highness and DNR degrees}
\affiliation{
  Department of Computer Science, National University of Ireland, Maynooth, co Kildare, Ireland.\\
Department of Mathematics, The National University of Singapore, 10 Lower Kent Ridge Drive, S17, Singapore 119076, Republic of Singapore.
}
\keywords{Bennett logical depth, Kolmogorov complexity, algorithmic randomness theory, computability and randomness.}
\received{2015-11-17}

\accepted{2017-10-1}

\publicationdetails{19}{2017}{4}{2}{1333}
\begin{document}
\maketitle
\begin{abstract}
  We study Bennett deep sequences in the context of recursion theory;
in particular we investigate the notions of $O(1)\deepk$,
$O(1)\deepc$, order$\deepk$ and order$\deepc$ sequences.
Our main results are that Mar\-tin-L\"of random sets are not
order$\deepc$, that every many-one degree contains a set which is
not $O(1)\deepc$, that $O(1)\deepc$ sets and order$\deepk$
sets have high or DNR Turing degree and that no $K$-trival
set is $O(1)\deepk$.
\end{abstract}
\section{Introduction}

\noindent
The concept of logical depth was introduced by C. Bennett \cite{b:bennett88}
to differentiate useful information (such as DNA) from the rest, with the key
observation that non-useful information pertains in both very simple
structures (for example, a crystal) and completely unstructured data (for
example, a random sequence, a gas). Bennett calls data containing useful
information logically deep data, whereas both trivial structures and
fully random data are called shallow.

The notion of useful information (as defined by logical depth)
strongly contrasts with classical information theory, which views
random data as having high information content. I.e., according to 
classical information theory, a random noise signal contains maximal
information,
whereas from the logical depth point of view, such a signal contains very
little useful information.

Bennett's logical depth notion is based on Kolmogorov complexity.
Intuitively a logically deep sequence (or equivalently a set) is one
for which the more
time a compressor is given, the better it can compress the sequence.
For example, both on trivial and random sequences, even when given more time,
a compressor cannot achieve a better compression ratio. Hence trivial and 
random sequences are not logically deep.

Several variants of logical depth have been studied in the past
\cite{b.antunes.depth.journal,DBLP:conf/cie/DotyM07,%
DBLP:journals/iandc/LathropL94,DBLP:journals/iandc/LathropL99,%
DBLP:journals/tcs/Moser13}.
As shown in \cite{DBLP:journals/tcs/Moser13}, all depth notions
proposed so far can be
interpreted in the compression framework which says a sequence
is deep if
given (arbitrarily) more than $t(n)$ time steps, a compressor
can compress the sequence $r(n)$ more bits than if given at 
most $t(n)$ time steps only. By considering different time bound families
for $t(n)$ (e.g. recursive, polynomial time etc.)
and the magnitude of compression improvement $r(n)$   - for
short: the \emph{depth magnitude} -
        (e.g. $O(1), O(\log n)$) one can capture all existing depth notions
\cite{b.antunes.depth.journal,DBLP:conf/cie/DotyM07,DBLP:journals/iandc/LathropL94,DBLP:journals/iandc/LathropL99,DBLP:journals/tcs/Moser13}
        in the compression framework \cite{DBLP:journals/tcs/Moser13}.
E.g. Bennett's notion is
        obtained by considering all recursive time bounds $t$ and a
constant depth magnitude,
        i.e., $r(n)=O(1)$.
        Several authors studied variants of Bennett's notion, by
considering different time bounds and/or different depth magnitude from
        Bennett's original notion
\cite{b.antunes.depth.journal,DBLP:journals/mst/AntunesMSV09,%
DBLP:conf/cie/DotyM07,DBLP:journals/iandc/LathropL94,%
DBLP:journals/tcs/Moser13}.

In this paper, we study the consequences these changes of different
parameters in 
Bennett's depth notion entail, by investigating the computational
power of the deep sets
yielded by each of these depth variants.
\begin{itemize}
\item We found out that the choice of the depth magnitude has consequences
 on the computational power of the corresponding deep sets.
 The fact that computational power implies Bennett depth was noticed
in \cite{DBLP:journals/iandc/LathropL94},
where it was shown that every high degree contains a Bennett deep set
(a set is high if, when given as an oracle,
its halting problem is at least as powerful as the halting problem
relative to the halting problem: $A$ is high iff $A' \geq_T \emptyset''$).
 We show that the converse also holds, i.e., that depth implies
computational power, by proving
 that if the depth magnitude is chosen to be ``large'' (i.e.,
$r(n)=\varepsilon n$), then
 depth coincides with highness (on the Turing degrees), i.e., a Turing
degree is high 
 iff it contains a deep set of magnitude $r(n)=\varepsilon n$.
\item For smaller choices of $r$, for example,
        if $r$ is any recursive order function, 
 depth still retains some computational power: we show that 
 depth implies either highness or diagonally-non-recursiveness,
denoted DNR (a total function is DNR if its image
 on input $e$ is different from the output of the $e$-th Turing
machine on input $e$). 
 This implies that if we restrict ourselves to
 left-r.e.\ sets, recursive order depth already implies highness.
 We also show that highness is not necessary by constructing
 a low order-deep set (a set is low if it is not powerful when given
as an oracle).
\item As a corollary, our results imply that weakly-useful sets
introduced in \cite{DBLP:journals/iandc/LathropL94} are either high
 or DNR (set $S$ is weakly-useful if the class of sets reducible to it
within a fixed time bound $s$
 does not have measure zero within the class of recursive sets).
\item Bennett's depth \cite{b:bennett88} is defined using prefix-free
Kolmogorov complexity.
 Two key properties of Bennett's notion are the so-called slow growth law, which
 stipulates that no shallow set can quickly (truth-table) compute a deep set,
 and the fact that neither Martin-L\"of random nor recursive  sets are deep.
 It is natural to ask whether replacing prefix-free with plain
complexity in Bennett's formulation
 yields a meaningful depth notion.
 We call this notion plain-depth.
 We show that the  random is not deep paradigm also holds in the setting
 of plain-depth. On the other hand we show that the slow growth law
fails for plain-depth:
 every many-one degree contains a set which is not plain-deep of
magnitude $O(1)$.
\item
A key property of depth is that ``easy'' sets should not be deep.
Bennett \cite{b:bennett88} showed that no recursive set is deep.
We give an improvement to this result by observing that no $K$-trivial
set is deep (a set
is $K$-trivial if the complexity of its prefixes is as low as possible).
Our result is close to optimal, since there exist deep 
ultracompressible sets \cite{DBLP:journals/iandc/LathropL99}.
\item
In most depth notions, the depth magnitude has to be achieved almost 
everywhere on the set. Some feasible depth notions also considered an
infinitely often version \cite{DBLP:conf/cie/DotyM07}. 
Bennett noticed in \cite{b:bennett88} that  infinitely often depth is
meaningless because every recursive set
is infinitely often deep.
We propose an alternative infinitely often depth notion that doesn't
suffer this limitation (called i.o.\ depth).
We show that little computational power is needed to compute i.o.\ depth,
i.e., every hyperimmune degree contains an i.o.\ deep set of magnitude
$\varepsilon n$ (a degree is hyperimmune
if it computes a function that is not bounded almost everywhere by
any recursive function),
and construct a $\Pi^0_1$-class
where every member is an i.o. deep set of magnitude $\varepsilon n$.
 For hyperimmune-free sets we prove that every non-recursive, non-DNR
hyperimmune-free set is i.o.\ deep of constant magnitude,
 and that every nonrecursive many-one degree contains such a set.
\end{itemize}
In summary, our results show that the choice of the magnitude for logical
depth has consequences on the computational
power of the corresponding deep sets, and that 
larger depth magnitude is not necessarily preferable over smaller magnitude.
We conclude with a few open questions regarding the constant magnitude case.

\section{Preliminaries}

We use standard computability/algorithmic randomness theory notations
see \cite{downey:book,nies:book,odifreddi}.
We use $\leq^+$ to denote less or equal up to a constant term. We fix
a recursive 1-1
pairing function $\langle \cdot \rangle : \N\times\N\rightarrow\N$.
We use sets and their characteristic sequences interchangeably,
we denote the binary strings of length $n$ by $\{0,1\}^n$
and $\{0,1\}^{\omega}$ denotes the set 
of all infinite binary sequences. The join of two sets $A,B$
is the set $A\oplus B$ whose characteristic sequence is
$A(0)B(0)A(1)B(1)\ldots$, that is, $(A \oplus B)(2n) = A(n)$ and
$(A \oplus B)(2n+1) = B(n)$ for all~$n$.
An order function is an unbounded non-decreasing function from
$\mathbb N$ to $\mathbb N$.
A time bound function is a recursive order $t$ such that there exists
a Turing machine
$\Phi$ such that for every $n$, $\Phi(n)[t(n)]\!\downarrow\, = t(n)$,
i.e., $\Phi(n)$ outputs
the value $t(n)$ within $t(n)$ steps of computation.
Set $A$ is left-r.e.\ iff the set of dyadic rationals strictly below
the real number $0.A$ (a.k.a. the left-cut of $A$ denoted $L(A)$) is 
recursively enumerable (r.e.), i.e., there is a recursive sequence of
non-decreasing rationals whose limit is $0.A$.
All r.e.\ sets are left-r.e., but the converse fails.

We consider standard Turing reductions $\leq_T$, truth-table
reductions $\leq_{tt}$ (where
all queries are made in advance and the reduction is total on all
oracles) and many-one reductions $\leq_m$.
Two sets $A,B$ are Turing equivalent ($A\equiv_T B$) if $A\leq_T B$
and $B\leq_T A$.
The Turing degree of a set $A$ is the set of sets Turing equivalent to $A$.
We fix a standard enumeration of all oracle Turing machines
$\Phi_1,\Phi_2,\ldots$;
the jump $A'$ of a set $A$ is the halting problem relative to $A$, i.e.,
$A' = \{e: \ \Phi^A_e(e)\!\downarrow\,\}$.
The halting problem is denoted $\emptyset'$.
A set $A$ is high (that is, has high Turing degree)
if its halting problem is as powerful as the halting problem of the halting
problem, i.e., $\emptyset'' \leq_T A'$.
High sets are equivalent to sets that compute dominating functions
(i.e., sets $A$ such that there is a function $f$ with $f\leq_T A$
such that for every computable function $g$ and for almost every $n$,
$f(n)\geq g(n)$), i.e., a set 
is high iff it computes a dominating function \cite{odifreddi}.
A set $A$ is low if its halting problem is not more powerful than the halting
problem of a recursive set, i.e., $A'\leq_T\emptyset'$. Note that
$\emptyset'$ is high relative to every low set.

If one weakens the dominating property of high sets to an infinitely
often condition, one obtains hyperimmune degrees.
A set is of hyperimmune degree if it computes a function that
dominates every recursive function on infinitely many inputs.
Otherwise the set is called of hyperimmune-free degree.

Another characterization of computational power used in computability
theory is the concept of diagonally non-recursive function (DNR).
A total function $g$ is DNR if for every $e$, $g(e)\neq \Phi_e(e)$,
i.e., $g$ can avoid the output of every Turing machine on at least one
specified input. A set is of DNR degree, if it computes a DNR
function. It is known that every r.e.\ DNR degree is high, actually even
Turing equivalent to $\emptyset'$ \cite{Arslanov}.

If one requires a DNR function to be Boolean, one obtains the PA-complete
degrees: A degree is PA-complete iff it computes a Boolean DNR function.
It is known that there exists low PA-complete degrees \cite{odifreddi}.

Fix a universal prefix free Turing machine $U$, i.e., 
such that no halting program of $U$ is a prefix of another halting program.
The prefix-free Kolmogorov complexity of string $x$, denoted $K(x)$,
is the length of the length-lexicographically first
program $x^*$ such that $U$  on input $x^*$ outputs $x$. It can be
shown that the value of $K(x)$  does not depend on the choice of 
$U$  up to an additive constant. 
$K(x,y)$ is the length of a shortest program that outputs the pair $\langle x,y\rangle$,
and $K(x|y)$ is the length of a shortest program such that $U$ outputs $x$
when given $y$ as an advice.
We also consider standard time bounded Kolmogorov complexity. 
Given time bound $t$ (resp. $s\in\N$), $K^t(x)$ (resp. $K_s(x)$)
denotes the length of the  shortest prefix free program $p$ such that
$U(p)$ outputs $x$ within $t(|x|)$ (resp. $s$) steps. 
Replacing $U$ above with a plain (i.e., non prefix-free) universal
Turing machine yields the notion
of plain Kolmogorov complexity, and is denoted $C(x)$.
We need the following counting theorem.

\begin{theorem}[Chaitin \cite{chaitin}]\label{t:counting}
There exists $c\in\N$ such that for every $r,n\in\N$, $|\{\sigma\in \{0,1\}^n:\
K(\sigma) \leq n + K(n) -r\}|\leq 2^{n-r+c}$.
\end{theorem}

\noindent
A set $A$ is Martin-L\"of random (MLR) if none of
its prefixes are compressible
by more than a constant term, i.e., $\forall n \ K(A\harp n) \geq n-c$
for some constant $c$, where $A\harp n$ denotes the first $n$ bits of
the characteristic function of $A$.
A set $A$ is $K$-trivial if its complexity is as low as possible, i.e.,
$\forall n \ K(A\harp n) \leq K(n)+O(1)$.
See the books of Downey and Hirschfeldt \cite{downey:book}
and Nies \cite{nies:book} for more on $C$ and $K$-complexity,
MLR and trivial sets.

Effective closed sets are captured by $\Pi^0_1$-classes.
A $\Pi^0_1$-class $P$ is a class of sequences such that there is a
computable relation $R$ such that
$P=\{S \in \{0,1\}^{\omega}| \ \forall n \ R(S\harp n) \}$.

\begin{adefinition}[Bennett \cite{b:bennett88}]
Let $g(n)\leq n$ be an order. A set $S$ is $g\deepk$ if for every recursive
time bound $t$ and for almost all $n\in\N$, $K^t(S\harp n) - K(S\harp
n) \geq g(n)$.
\end{adefinition}

\noindent
A set $S$ is $O(1)\deepk$ (resp.\ order$\deepk$) if it is $c\deepk$
(resp.\ $g\deepk$) for every $c \in\mathbb N$ (resp.\ for some recursive
order $g$).
A set is said Bennett deep if it is $O(1)\deepk$.
We denote by $g\deepc$ the above notions with $K$ replaced with $C$.
It is easy to see that for every two orders $f,g$ such that $\forall
n\in\N \ f(n)\leq g(n)$,
every $g\deepk$ set is also $f\deepk$.

Bennett's slow growth law (SGL) states that creating depth requires 
time beyond a ``recursive amount'',
i.e., no shallow set quickly computes a deep one.

\begin{alemma}[Bennett \cite{b:bennett88}; Juedes, Lathrop and Lutz
\cite{DBLP:journals/iandc/LathropL94}]\label{l.sgl}
Let $h$ be a recursive order, and $A\leq_{tt} B$ be two sets.
If $A$ is $h\deepk$ (resp.\ $O(1) \deepk$) then $B$ is $h' \deepk$
(resp.\ $O(1) \deepk$) for some recursive order $h'$.
Furthermore given indices for the truth-table reduction and for $h$,
one can effectively compute an index for $h'$.
\end{alemma}

\noindent
The symmetry of information holds in the resource bounded case.

\begin{alemma}[Li and Vit\'anyi \cite{Li2008}]\label{l.sym.inform.}
For every time bound $t$, there is a time bound $t'$ such that for all
strings $x,y$ 
with $|y|\leq t(|x|)$, we have $C^t(x,y) \geq C^{t'}(x) + C^{t'}(y\ |
\ x) - O(\log C^{t'}(x,y))$.
Furthermore given an index for $t$ one can effectively compute an
index for $t'$.
\end{alemma}

\begin{acorollary}\label{c.sym.inform.}
Let $t$ be a time bound and $x,a$ be strings. Then there exists a time
bound $t'$ such that 
for every prefix $y$ of $x$ we have $C^t(y \ | \ a) \geq^+ C^{t'}(x \
|\  a) - |x| + |y| - O(\log C^t(y \ | \ a))$.
Furthermore given an index for $t$ one can effectively compute an
index for $t'$.
\end{acorollary}

\begin{proof}
Given $p$ a $C^t$-minimal program with advice $a$ for $y$, $|x|-|y|$
remaining bits, and a delimiter after $p$, 
one can reconstruct $x$ in time $t'(\cdot)=t(\cdot) + O(n)$ steps given a. Thus
$C^t(y\ | \ a ) + |x| - |y| +O(\log C^t(y\ | \ a)) \geq^+ C^{t'}(x\ | \ a )$.
\end{proof}

\section{$C$-Depth}

\noindent
Bennett's original formulation \cite{b:bennett88} is based on $K$-complexity.
In this section we investigate the depth notion obtained by replacing
$K$ with $C$, which we call plain depth.
We study the interactions of plain depth with the notions of
Martin-L\"of random sets, many-one
degrees and the Turing degrees of deep sets.

\subsection{MLR is not order$\deepc$}

\noindent
The following result is the plain complexity version of
Bennett's result that no MLR sets are Bennett deep.

\begin{theorem} \label{mlrdeepc}
For every MLR set $A$ and for every recursive order $h$,
$A$ is not $h \deepc$.
\end{theorem}

\begin{proof}
Suppose by contradiction that set $A$ is MLR and $h \deepc$, for some
$h$ as above.
We claim that $\exists^{\infty} n \ C(A\harp n) \geq n - h(n)/2$. To
prove the claim, let 
$N=\{n\in\N:\ h(n)\neq h(n-1)\}$. Then given $a=h(n)$ with $n\in N$,
the program $p$: ``Print 
the smallest $i$ such that $h(i)=a$.'' is a program for $n$ of size
$K(a)+O(1)$, i.e.,
$$K(n) \leq K(a) +O(1) \leq 2 \log h(n) +O(1) < h(n)/4. $$
Suppose $q$ is a $C$-minimal program for $A\harp n$ of size $n-m$,
then appending
$2\log m + K(n)$ bits to $q$ yields a 
prefix free program $q'$ for $A\harp n$ of size $2\log m + K(n) + n -m$.

Since $A$ is MLR we have $|q'|\geq n -O(1)$, i.e., $2\log m + K(n) + n
-m >n -O(1)$ which implies
$K(n) > m - 2\log m -O(1) > 2/3 m$ (for $m$ sufficiently large).
If $n\in N$, then $2/3m < K(n) < h(n)/4$, thus $m<h(n)/2$, i.e.,
$C(A\harp n) > n - h(n)/2$,
which proves the claim.

Since for all $n\in\N$ we have $C^{n^2}(A\harp n) \leq n + O(1)$ (via
a ``print'' program), it follows that for every $n\in N$,
$$C^{n^2}(A\harp n) - C(A\harp n) \leq n + O(1) - n + h(n)/2 \leq^+ h(n)/2$$
which contradicts that $A$ is $h \deepc$.
\end{proof}

\medskip
\noindent
Sequences that are MLR relative to the halting problem are called $2$-random.
Equivalently a sequence $A$ is   $2$-random iff there is a constant
$c$ such that $C(A \harp n) \geq n-c$ for infinitely many $n$
\cite{Mi04,NST05}. Since there is a constant $c'$ such that $n+c'$ is a trivial
upper bound on the plain Kolmogorov complexity of any string of length $n$,
it is clear that no $2$-random sequence can be $O(1)\deepc$. Thus most
MLR sequences
are not $O(1)\deepc$.

\subsection{The SGL fails for $C$-depth}

\noindent
The following result shows that 
the Slow Growth Law fails for plain depth.

\begin{theorem}\label{t:manyone-deg-deepC}
Every many-one degree contains a set which is not $O(1) \deepc$.
\end{theorem}

\begin{proof}
The recursive many-one degrees consist only of sets which are
not $O(1) \deepc$. So consider any set
$A$ different from both $\emptyset$ and $\mathbb N$
and let $B = \{2^{2^p}: p \in A\}$.  Given any $k$,
choose $m = A(0)+2 A(1)+\ldots+2^k A(k)$ and let $n$ be any number
between $2^m$ and $2^{m+1}$ which has $C$-complexity $m$. Now
on one hand $C(B \upharpoonright n) \geq^+ C(n) \geq m$ and
on the other hand, one can compute $m$ from the $m$-digit binary number
representing $n$ and one can compute $A(0),A(1),\ldots,A(k)$ from $m$
and using $B(2^{2^p}) = A(p)$ one can compute $B \upharpoonright n$
from the binary representation of $n$ and its length $m$ so that
$C^t(B \upharpoonright n) \leq m+c$ for some time bound $t$ and some
constant $c$ independent of $n$
and $m$. This shows that $A$ is not $O(1) \deepc$. Clearly $A \leq_m B$.
Furthermore, $B \leq_m A$ by mapping all values of form $2^{2^p}$ to $p$ and
all other values to a fixed non-element of $A$.
\end{proof}

\medskip
\noindent
Note that this result shows that order$\deepk$ does not imply
order$\deepc$: all the sets in the  truth-table degree of any order$\deepk$
set are all order$\deepk$ (by the SGL), but this degree contains
a non order$\deepc$ set by the previous result.

\subsection{Depth implies highness or DNR}

\noindent
The following result shows that being constant deep for $C$ 
implies computational power.

\begin{theorem}\label{t.deep.high.or.dnc}
Let $A$ be an $O(1) \deepc$ set.  Then $A$ is high or DNR.
\end{theorem}

\begin{proof}
We prove the contrapositive.
Suppose that $A$ is neither DNR nor high. Let $f(m)$ be (a coding of)
$A \upharpoonright 2^{m+1}$. Because $f\leq_{tt}A$, there are infinitely many
$m$ where $\Phi_m(m)$ is defined and equal to $f(m)$. Hence there
is an $A$-recursive increasing function $g$
such that, for almost every  $m$, $g(m)$ is the time to find an $m' \geq m$
with $A(0)A(1) \ldots A(2^{m'+1}) = \Phi_{m'}(m')$ and to evaluate the
expression $\Phi_{m'}(m')$ to verify the finding. As $A$ is not high,
there is a recursive increasing function $h$ with $h(m) \geq g(m)$ for
infinitely many $m$. Now consider any $m$ where $h(m) \geq g(m)$. Then
for the $m'$ found for this $m$, it holds that $h(m') \geq h(m)$ and
$h(m')$ is also larger than the time to evaluate $\Phi_{m'}(m')$.
Hence $h(m)$ is larger than the time to evaluate $\Phi_m(m)$ for
infinitely many $m$ where $\Phi_m(m)$ codes $A \upharpoonright 2^{m+1}$.

For each such $m$, let $n$ be a number with $2^m \leq n \leq 2^{m+1}$
and $C(n) \geq m$. Starting with a binary description of such an $n$,
one can compute $m$ from $n$ and run $\Phi_m(m)$ for $h(m)$ steps
and, in the case that this terminates with a string $\sigma$ of length
$2^{m+1}$, output $\sigma \upharpoonright n$. It follows from this algorithm
that there is a resource-bounded approximation to $C$ such that
there exist infinitely many $n$ such that, on one hand
$C(A \upharpoonright n) \geq \log(n)$ while on the
other hand $A \upharpoonright n$ can be described in $\log(n)+c$ bits
using this resource bounded description. Hence $A$ is not $O(1) \deepc$. 
\end{proof}

\medskip
\noindent
Since there are incomplete r.e.\ Turing degrees which are high, these are
also not DNR and, by Theorem~\ref{t.high.iff.deep}, they contain sets
which are $0.9\deepc$. Thus the preceeding theorem cannot be improved to
show that ``$O(1)\deepc$ sets are DNR''.

\begin{theorem}\label{t:deep-not-DNR}
There exists a set $A$ such that $A$ is $(1-\varepsilon)n \deepc$
(for any $\varepsilon >0$) but  $A$ is not DNR.
\end{theorem}

\begin{proof}
There is a degree which is high but not DNR \cite{odifreddi}.
Thus we can, by Theorem \ref{t.high.iff.deep}, select a set $A$
in this degree which is $(1-\varepsilon)n \deepc$ for every $\varepsilon < 1$.
\end{proof}

\section{$K$-Depth}

\noindent
Bennett's original depth notion is based on prefix free complexity.
He made important connections between depth and truth-table degrees;
In particular he proved that the $O(1)\deepk$ sets are
closed upward under truth-table reducibility, which he called the slow
growth law. 
In the following section we 
pursue Bennett's investigation by studying the Turing degrees of deep sets.
In the first subsection, we investigate the connections
between linear depth and high Turing degrees. We then
look at the opposite end by studying the interactions of  various
lowness notions
with logical depth.

\subsection{Highness and depth coincide}

\noindent
The following result shows that at depth magnitude $\varepsilon n$,
depth and highness coincide on the Turing degrees.
The result holds for both $K$ and $C$ depth.

\begin{theorem}\label{t.high.iff.deep}
For every set $A$ the following statements are equivalent:
\begin{enumerate}
\item The degree of $A$ is $\varepsilon n \deepc$ for some $\varepsilon > 0$.
\item The degree of $A$ is $(1-\varepsilon)n \deepc$ for every $\varepsilon >0$.
\item $A$ is high.
\end{enumerate}
\end{theorem}

\begin{proof}
We prove $(1) \Rightarrow (3)$ using the contrapositive: 
Let $\varepsilon>0$ and $l\in\N$ such that $\delta < \varepsilon/3$
with $\delta:= 1/l$. Let $k$ be the limit inferior
of the set $\{0,1,\ldots,l\}$ such that there are infinitely many
$n$ with $C(A \upharpoonright n) \leq n \cdot k \cdot \delta$. Now one can
define, relative to $A$, an $A$-recursive function $g$ such that
for each $n$ there is an $m$ with $n \leq m \leq g(n)$ and
$C_{g(n)}(A \upharpoonright m) \leq m \cdot k \cdot \delta$.
As $A$ is not high, there is a recursive function $h$ with $h(n) > g(n)$
for infinitely many $n$; furthermore, $h(n) \leq h(n+1)$ for all~$n$.
It follows that there are infinitely many $n$ with
$C_{h(n)}(A \upharpoonright n) \leq n \cdot k \cdot \delta$ which
is also at most $n \cdot \delta$ away from the optimal value, hence $A$
is not $\varepsilon \cdot n$ deep, which ends this direction's proof.

Let us show $(3)\Rightarrow (2)$.
Let $\varepsilon > 0$, $A$ be high,   and 
let $g\leq_T A$ be dominating. We construct $B\equiv_T A$ such that
$B$ is $(1-\varepsilon)n \deepc$.

By definition, if $t$ is a time bound and $i$ an index of $t$ then
for every $m\in\N$ $\Phi_i(m)[t(m)]\!\downarrow\, = t(m)$.
Since $g$ is dominating, we have
for almost  every $m\in\N$, $t(m) = \Phi_{i}(m)[g(m)]\!\downarrow$.

We can thus use $g$ to encode all time bounds that are total on all
strings of length less than a certain bound into a set $H$, where
\begin{quote}
$H(\langle i,j\rangle) = 1 $ iff  $\Phi_{i}(m)[g(2^j)]\!\downarrow$ for
all $m\in\{1,2,\ldots,2^j\}$.
\end{quote}
Thus $t$ is a (total) time bound iff for almost every $j$, $H(\langle
i,j\rangle) = 1 $ (where $i$ is an index for $t$).

We have $H\leq_T A$ and we choose the pairing function $\langle \cdot \rangle$
such that $H\upharpoonright n^2+1$ encodes the values
$$
   \{H(\langle i,j\rangle): \ i,j \leq n\}.
$$
Let $n\in\N$ and suppose  $B\upharpoonright 2^n$ is already
constructed. Given $A\upharpoonright n+1$ and $H\upharpoonright n^2+1$,
we construct $B\upharpoonright 2^{n+1}$.
From $H\upharpoonright n^2+1$, we can compute the set $L_n=\{i\leq n:
\  H(\langle i,n\rangle)=1\}$, i.e.,
a list eventually containing all time bounds that are total on strings
of lengths less or equal to $2^n$.
Let $$T_n := \max\{\Phi_i(m): \ i\in L_n, m\leq 2^n\}.$$
Find the lex first string $x_n$ of length $2^n -1$ such that
$$C_{T_n} (x_n \ | \ (B\upharpoonright 2^n) A(n))\geq 2^n.$$ 
Let $ B\upharpoonright [2^n,2^{n+1}-1] :=  A(n) x_n.$
By construction we have $B\equiv_T A$. Also,
$C(B\upharpoonright 2^{n+1}\ | \ H\upharpoonright n^2+1,
A\upharpoonright n+1) \leq^+ C(n)$, i.e.,
$C(B\upharpoonright 2^{n+1})\leq 2n^2$.

Let us prove $B$ is $\frac{2}{3}n \deepc$; we then extend the argument
to show $B$ is $(1-\varepsilon)n \deepc$.
Let $t$ be a time bound. Let $n$ be large enough such that
$t_1,t_2,t_3 \in L_{n-2}$ and $t_1',t_2',t_3',t'_4 \in L_{n}$
where the $t_i$'s  are derived from $t$ as described below.

Let $j$ be such that $2^{n} < j \leq 2^{n+1}$ and $j'=j-(2^{n}-1)$,
i.e., $B\upharpoonright j$ ends
with the first $j'-2$ bits of $x_n$ (One bit is ``lost'' due to the
first bit used to encode $A(n)$).

We consider two cases, first suppose $j'<\log n$. Let $t_1$ be a time
bound (obtained from $t$) such
that $C^t(B\upharpoonright j) \geq^+ C^{t_1} (x_{n-1},B\upharpoonright
2^{n-1})$, where neither the constant nor $t_1$ depends
on $j,n$. Let $t_2$ be derived from $t_1$ using Lemma \ref{l.sym.inform.}.
We have
\begin{align*}
 & C^{t_1} (x_{n-1},B\upharpoonright 2^{n-1}) & \\
        &\geq C^{t_2} (B\upharpoonright 2^{n-1}) + C^{t_2} (x_{n-1}\ |
\ B\upharpoonright 2^{n-1}) - O(\log 2^n)& \\
 &\geq C^{t_2} (B\upharpoonright 2^{n-1}) + C_{T_{n-1}} (x_{n-1}\ | \
B\upharpoonright 2^{n-1}) - O(n) & \text{because $t_2\in L_{n-1}$}\\
 &\geq C^{t_2} (B\upharpoonright 2^{n-1}) + 2^{n-1} - O(n) & \text{by
definition of $x_{n-1}$} \\
 &\geq 2^{n-1} + 2^{n-2} + C^{t_3} (B\upharpoonright 2^{n-2}) - O(n) &
\text{reapplying the argument above} \\
 &\geq \frac{3}{4} 2^n -O(n) > \frac{2}{3}(2^n +j'+1) = \frac{2}{3}  j. & 
\end{align*}
For the second case, suppose $j'>\log n$. We have
\begin{align*}
 C^t(B\upharpoonright j) & \geq  
 C^{t'_1} (x_{n}\upharpoonright j',B\upharpoonright 2^{n}) &\\
 &\geq C^{t'_2} (x_n\upharpoonright j'\ | \ B\upharpoonright 2^n) +
C^{t'_2} (B\upharpoonright 2^{n}) - O(n) &\text{By Lemma
\ref{l.sym.inform.}} \\
 &\geq C^{t'_3} (x_n| \ B\upharpoonright 2^n) -2^n +j' + C^{t'_2}
(B\upharpoonright 2^{n}) - O(n) &\text{by Corollary \ref{c.sym.inform.}}\\
 &\geq C_{T_n} (x_n\ | \ B\upharpoonright 2^{n}) -2^n + j'  + C^{t'_2}
(B\upharpoonright 2^{n})- O(n) &\ \text{because $t'_3\in L_{n}$}\\
 &\geq 2^{n} -2^n + j'  + C^{t'_2} (B\upharpoonright 2^{n})- O(n)
&\text{by definition of $x_n$} \\
 &\geq j'  + C^{t'_4} (x_{n-1},B\upharpoonright 2^{n-1})- O(n) &
\text{same as in the first case}\\
 &\geq j' +\frac{3}{4} 2^n -O(n) &\text{same as in the first case}\\
 &> \frac{2}{3} j &
\end{align*}
Note that each iteration of the argument above yields a $2^{n-k}$ term
($k=1,2,3,\ldots$), therefore for any $\varepsilon >0$,
there is a number $I$ of iterations, such that $B$ can be shown
$(1-\varepsilon)n \deepc$, for all $n$ large enough such
that $t_1,t_2,\ldots,t_{3I} \in L_n$. 
\end{proof}

\begin{acorollary}\label{c:Kalso-holds}
Theorem \ref{t.high.iff.deep} also holds for $K$-depth.
\end{acorollary}

\begin{proof}
Because every set $A$ is $\varepsilon n \deepc$ (for some $\varepsilon
>0$) iff it is $\varepsilon' n \deepk$ for some $\varepsilon' >0$,
since for every $x$, $C(x) \leq K(x) \leq C(x) + O(\log |x|)$.
\end{proof}

\subsection{Depth implies highness or DNR}

\noindent
An analogue of Theorem \ref{t.deep.high.or.dnc} holds for $K$.

\begin{theorem}\label{t.deep.high.or.dnc.k}
Let $A$ be a  $h \deepk$ set for some recursive order $h$.  Then $A$
is high or DNR.
\end{theorem}

\begin{proof}
We prove the contrapositive.
Suppose that $A$ is neither DNR nor high. Let $f(m)$ be (a coding of)
$A \upharpoonright h^{-1}(m)$, where $h^{-1}(m)=\min_n \{h(n)=m\}$.
The rest of the proof follows the proof of Theorem \ref{t.deep.high.or.dnc},
with $2^m$ replaced with $h^{-1}(m)$.
\end{proof}

\medskip
\noindent
As a corollary, we show that in the left-r.e.\ case, depth always
implies highness.

\begin{acorollary}\label{c:l-re-then-high}
If $A$ is left-r.e.\ and $h \deepk$ (for some recursive order $h$)
then $A$ is high.
\end{acorollary}

\begin{proof}
Let $A$ be as above. By definition of  $A$ being left-r.e., the left
cut $L(A)$  of $A$ is r.e.\ and
$L(A)\equiv_{tt} A$. By Lemma \ref{l.sgl},
$L(A)$ is $h' \deepk$ (for some recursive order $h'$).
By Theorem \ref{t.deep.high.or.dnc.k}, $L(A)$ is high or DNR. Since
every r.e.\ DNR set is high, $A$ is high.
\end{proof}

\medskip
\noindent
As a second corollary, we prove that every weakly-useful set
is either high or DNR. A set $A$ is weakly-useful
if there is a time-bound $s$ such that the class of all sets
truth-table reducible to $A$ with this time bound $s$ is not small,
i.e., does not have measure zero within the class of recursive sets;
see \cite{DBLP:journals/iandc/LathropL94} for a precise definition.
In \cite{DBLP:journals/iandc/LathropL94}, it was shown that every
weakly-useful set is $O(1)\deepk$ (even order$\deepk$
as observed in \cite{DBLP:journals/mst/AntunesMSV09}) thus generalising the
fact that $\emptyset'$ is $O(1)\deepk$, since $\emptyset'$ is weakly-useful.

\begin{theorem}[Antunes, Matos, Souto and 
Vit{\'a}nyi\cite{DBLP:journals/mst/AntunesMSV09};
Juedes, Lathrop and Lutz \cite{DBLP:journals/iandc/LathropL94}]\label{t:wudeep}
Every weakly-useful set is order$\deepk$.
\end{theorem}

\noindent
It is shown in \cite{DBLP:journals/iandc/LathropL94} that every high
degree contains a weakly-useful set.
Our results show some type of converse to this fact.

\begin{theorem}\label{t:wu-then-hi}
Every weakly-useful set is either high or DNR.
\end{theorem}

\begin{proof}
This follows from Theorem \ref{t.deep.high.or.dnc.k},
since every weakly-useful set is order$\deepk$ by Theorem \ref{t:wudeep}.
\end{proof}

\subsection{A low deep set}

\noindent
We showed in Theorem \ref{t.high.iff.deep} that every $\varepsilon
n\deepk$ set is high.
Also Theorem \ref{t.deep.high.or.dnc.k} shows that every order$\deepk$
set is either high or DNR.
Thus one might wonder whether there exists any non-high order$\deepk$ set.
We answer this question affirmatively by showing there exist low
order$\deepk$ sets.

\begin{theorem}\label{t:pawu}
If $A$ has PA-complete degree, then there exists a weakly-useful set
$B\equiv_T A$.
\end{theorem}

\begin{proof}
Let $f\leq_T A$ be a Boolean DNR function and let $g(n):= 1-f(n)$.
It follows that if $\Phi_e$ is Boolean and total, then $g(e)=\Phi_e(e)$.
One can thus encode $g$ into a set $B\leq_T A$ such that 
for every $e$ such that $\Phi_e$ is Boolean and total and for every $x$,
$B(\langle e+1,x \rangle) = \Phi_e(x)$. One can also encode $A$ into $B$
(for example, $B(\langle 0,x\rangle)=A(x)$) so that $A\equiv_T B$.
Thus for every recursive set $L$ there exists $e$ such that  for every
string $x$, we have $L(x)=B(r_e(x))$,
where $r_e(x)=\langle e,x\rangle$ is computable within $s(n)=n^2$
steps (by using a lookup table on small inputs).
It follows that every recursive set is truth-table reducible to $B$
within time $s(n)=n^2$.
Because the class of recursive sets does not have measure zero within
the class of recursive sets \cite{DBLP:journals/iandc/LathropL94},
it follows that $B$ is weakly-useful.
\end{proof}

\begin{acorollary}\label{c:padeep}
If $A$ has PA-complete degree, then there exists an order$\deepk$  set
$B\equiv_T A$.
Furthermore, there is a $\Pi^0_1$-class only consisting of
order$\deepk$ sets.
\end{acorollary}

\noindent
This corollary follows from Theorems \ref{t:wudeep} and \ref{t:pawu}.
Recall that a set $A$ is said low for $\Omega$ iff Chaitin's $\Omega$ is 
Martin-L\"of random relative to $A$; a set $A$ has superlow degree if
its jump $A'$ is truth-table reducible to the halting problem.
Well-known basis theorems \cite{downey:low.omega,joso} yield the following
corollary.

\begin{acorollary} \label{c:low-deep}
For each of the properties low, superlow, low for $\Omega$
and hyperimmune-free, there exists an order$\deepk$ set which
also has this respective property.
\end{acorollary}

\begin{proof}
There exists low sets $A$ of PA-complete degree \cite{odifreddi}.
By Theorem \ref{c:padeep} there exists an  order$\deepk$ set $B\equiv_T A$ .
Since $A$ is low it follows that $B$ is low.
\end{proof}

\medskip
\noindent
\noindent

The reason one uses PA-complete sets instead of merely Martin-L\"of random sets
(which also satisfy all basis theorems),
 is that Martin-L\"of random sets are not
weakly-useful; indeed, it is known that they are not even $O(1)\deepk$.
This stands in contrast to the following result.

\begin{acorollary}
There are two Martin-L\"of random sets $A$ and $B$ such that
$A \oplus B$ is order$\deepk$.
\end{acorollary}

\begin{proof}
Barmpalias, Lewis and Ng \cite{BLN10} showed that every PA-complete
degree is the join of two Martin-L\"of random degrees; hence there are
Martin-L\"of random sets $A,B$ such that $A \oplus B$ is a hyperimmune-free
PA-complete set. Thus, by Theorem~\ref{t:pawu} there is a weakly-useful
set Turing reducible to $A \oplus B$ which, due to the hyperimmune-freeness,
is indeed truth-table reducible to $A \oplus B$. It follows that $A \oplus B$
is itself weakly-useful and therefore order$\deepk$ by
Theorem~\ref{t:wudeep}.
\end{proof}

\subsection{No $K$-trivial is $O(1)\deepk$}

\noindent
A key property of depth is that ``easy'' sets should not be deep.
Bennett \cite{b:bennett88} showed that no recursive set is deep.
Here we improve this result by observing that no $K$-trivial set is deep.
It  follows easily from equivalent characterisations of $K$-triviality (see \cite{downey:book,nies:book}),
but our proof is self-contained.
As we will see this result is close to optimal. 

\begin{theorem}\label{t:ktshallow}
No $K$-trivial set is $O(1)\deepk$.
\end{theorem}

\begin{proof}
%
Let $A$ be $K$-trivial and $c\in\N$ such that $\forall n\in\N$,
$K(A\harp n)\leq K(n) +c$. Let $d$ be such that
for every string $x$, $K(x)\geq K(|x|) -d$ and let $g(n)=n^2$.
There exists a constant $d'$ such that the set $M= \{n\in\N: \ K^g(n)
\leq K(n) + d'\}$ is infinite (see \cite{downey:book} p. 139). 
Note that $M$ is co-r.e., i.e., there exists uniformly recursive
approximations $M_1\supseteq M_2 \supseteq \ldots \supseteq M$ of $M$.
Let
$c'= \liminf_{n\in M} |\{\sigma\in \{0,1\}^n: \ K(\sigma)\leq K^g(n) + c\}|$.
By Theorem \ref{t:counting} (with $r=n+K(n)-K^g(n)-c$), $c' < \infty$.
Consider the function $$f(n) = \min_s \{n\not\in M_s \text{ or there
exist } c' \text{ strings }  \sigma\in \{0,1\}^n \text{ with }
K_s(\sigma) \leq K^g(n) +c \}.$$
By modifying $f$ on the finitely many values before the liminf is
reached, $f$ is recursive.
Wlog $f$ is bounded by a time bound which we also denote $f$.
We have $\exists^{\infty} n\in M$ such that $K^f(A\harp n) \leq K(n) +
c+d'$ thus for each of these infinitely many $n$'s
we have $K^f(A\harp n) - K(A\harp n) \leq K(n) + c +d' -K(n) + d =
c+d+d'$, i.e., $A$ is not $O(1)\deepk$.
\end{proof}

\medskip
\noindent
Call a set $A$ ultracompressible if for every recursive order $g$
and all $n$, $K(A\harp n) \leq^+ K(n)+g(n)$.
The following theorem shows that our result is close to optimal.

\begin{theorem}[Lathrop and Lutz \cite{DBLP:journals/iandc/LathropL99}]
There is an ultracompressible set $A$
which is $O(1)\deepk$.
\end{theorem}

\begin{theorem}[Herbert \cite{herbertdiss}]
There is a set $A$ which is not $K$-trivial but which satisfies that
for every $\Delta^0_2$ order $g$ and all $n$,
$K(A\harp n) \leq^+ K(n)+g(n)$.
\end{theorem}

\noindent
It would be interesting to know whether such sets as found by Herbert
can be $O(1)\deepk$. The result of Herbert is optimal,
Csima and Montalb\'an \cite{csima-montalaban:min-pair-K-deg}
showed that such sets do not exist when using $\Delta^0_4$ orders
and Baartse and Barmpalias \cite{baartse-barmpalias} improved this
non-existence to the level $\Delta^0_3$. We also point to related
work of Hirschfeldt and Weber \cite{HW12}.

\begin{theorem}[Baartse and Barmpalias \cite{baartse-barmpalias}]
There is a $\Delta^0_3$ order $g$ such that a set $A$ is $K$-trivial iff
$K(A\harp n) \leq^+ K(n) + g(n)$ for all $n$.
\end{theorem}

\section{Infinitely Often Depth and Conditional Depth}

\noindent
Bennett observed in \cite{b:bennett88}
that being infinitely often Bennett deep is meaningless, because
all recursive sets are infinitely often deep.
A possibility for a more meaningful notion of  infinitely often depth,
is to consider 
a depth notion where the length of the input is given as an advice. We
call this notion i.o.\ depth.

\begin{adefinition}
A set $A$ is i.o.\ $O(1) \deepk$ if for every $c \in\N$ and for every
time bound $t$ there are infinitely many $n$ satisfying
$K^{t}(A\harp n\ | \ n) - K(A\harp n\ | \ n) \geq c.$
\end{adefinition}

\noindent
If we replace $K$ with $C$ in the above definition, we call the
corresponding notion i.o.\ $O(1) \deepc$.
The fact that all recursive sets are infinitely often deep in Bennett's
approach does no longer hold for i.o.\ depth as defined above.

\begin{alemma}
Let $A$ be recursive. Then $A$ is neither i.o.\ $O(1) \deepc$
nor i.o.\ $O(1) \deepk$.
\end{alemma}

\begin{proof}
Let $A$ be recursive and $t$ be a time bound.
Wlog $A$ is recursive in time $t$, i.e., for every $n\in\N$ we have
$C^t(A\harp n\ | \ n) \leq c$ for some constant $c$, thus
$\forall n\ C^{t}(A\harp n\ | \ n) - C(A\harp n\ | \ n) < c$.
The $K$ case is similar.
\end{proof}

\medskip
\noindent
The following shows that very little computational power is needed to
compute an i.o.\ deep set.

\begin{theorem}\label{t:iodeep-hypimmune}
\begin{enumerate}
\item There is a $\Pi^0_1$-class such that every member is
      i.o.\ $\varepsilon n \deepc$ for all $\varepsilon < 1$.
      In particular there is such a set of hyperimmune-free degree.
      Furthermore, every hyperimmune Turing degree contains such a set.
\item Every nonrecursive many-one degree contains an
      i.o.\ $O(1) \deepc$ set.
\item If $A$ is neither recursive nor DNR, then $A$
      is i.o.\ $O(1) \deepc$.
\end{enumerate}
\end{theorem}

\begin{proof}
This result is obtained by splitting the natural numbers recursively into
intervals $I_n = \{a_n,\ldots,b_n\}$ such that $b_n = (2+a_n)^2$.
Now one defines the $\Pi^0_1$-class such that for each
$n = \langle e,k \rangle$ where $t = \varphi_e$ is defined up to $b_n$,
a string $\tau \in \{0,1\}^{b_n-a_n+1}$ is selected such that for all
$\sigma \in \{0,1\}^{a_n}$, $C^t(\sigma\tau) \geq b_n-2a_n-2$
and then it is fixed that all members $A$ of the $\Pi^0_1$-class have
to satisfy $A(x) = \tau(x-a_n)$ for all $x \in I_n$.
Since there are $2^{b_n-a_n+1}$ strings $\tau$ and for each program of
size below $b_n-2a_n-2$ can witness that only $2^{a_n}$ many
$\tau$ are violating $C^t(\sigma\tau) \geq |\tau|-|\sigma|$ for some
$\sigma \in \{0,1\}^{a_n}$, there will be less than
$2^{b_n-a_n+1}-2^{b_n-a_n}$ many $\tau$ that get disqualified and so
the search finds such a $\tau$ whenever $\varphi_e$ is defined
up to $b_n$. Hence, for every total $t = \varphi_e$, there are
infinitely many intervals $I_n$ with $n$ of the form $\langle e,k \rangle$
such that on these $I_n$, $C^t(A(0)A(1)\ldots A(b_n)\,|\,n) \geq
C^t(A(0)A(1)\ldots A(b_n)) - \log(n) \geq  b_n-3a_n$
and $C(A(0)A(1)\ldots A(b_n)|n) \leq a_n+c$ for a constant $c$,
as the program only needs to know how $A$ behaves below $a_n$
and can fill in the values of $\tau$ on $I_n$. So the complexity
improves after time $t(b_n)$ from $b_n-3a_n$ to $a_n$ and,
to absorb constants, one can conservatively estimate the
improvement by $b_n-5a_n$.
By the choice of $a_n,b_n$, the ratio $(b_n-5a_n)/b_n$ tends
to $1$ and therefore every $A$ in the $\Pi^0_1$-class is
$\varepsilon n\deepc$ for every $\varepsilon < 1$.
Note that there are hyperimmune-free sets inside this
$\Pi^0_1$-class, as it has only nonrecursive members.

Furthermore, one can see that the proof also can be adjusted
to constructing a single set in a hyperimmune Turing degree
rather than constructing a full $\Pi^0_1$-class. In that case
one takes some function $f$ in this degree which is not dominated
by any recursive function and then one permits for each
$n = \langle e,k\rangle$ the time $\varphi_e(b_n)$ in the case
that $\varphi_e(b_n) < f(k)$ and chooses $\tau$ accordingly
and one takes $\tau = 0^{b_n-a_n+1}$ in the case that $\varphi_e$
does not converge on all values below $b_n$ within time
$f(k)$ otherwise. This construction is recursive in the given
degree and a slight modification of this construction would permit
to code the degree into the set $A$.

\medskip
\noindent
For the second item, consider a set $A \subseteq \{4^n: n \in {\mathbb N}\}$.
Every many-one degree contains such a set. For each binary string $\sigma$,
let
\begin{quote}
$S_\sigma = \{\tau \in \{0,1\}^*: 4^{|\sigma|-1} < |\tau| \leq 4^{|\sigma|}$
and $\tau(4^n) = \sigma(n)$ for all $n < |\sigma|$ and $\tau(n) = 0$
for all $n < |\tau|$ which are not a power of $4\}$.
\end{quote}
In other word, for every $A \subseteq \{4^n: n \in {\mathbb N}\}$,
$S_{A(1)A(4)A(16)\ldots A(4^n)}$ contains those $\tau$ which are a
prefix of $A$ and for which $\tau(4^n)$ is defined but not $\tau(4^{n+1})$.
For each $e,k,n$ where $\varphi_e$ is a total function $t$,
we now try to find inductively for $m = 4^n+1,4^n+2,\ldots,4^{n+1}$
strings $\sigma_m \in \{0,1\}^{n+1}$ such that
whenever $\sigma_m$ is found then it is different from all those
$\sigma_{m'}$ which have been found for some $m'<m$
and the unique $\tau \in S_{\sigma_m} \cap \{0,1\}^m$
satisfies $C^t(\tau\,|\,m) \geq e+3k$. Note that due to the
resource-bound on $C^t$ one can for each $m'<m$ check whether
$\sigma_{m'}$ exists and take this information into account
when trying to find $\sigma_m$. Therefore, for those $m$ where
$\sigma_m$ exists, the $\tau \in S_{\sigma_m} \cap \{0,1\}^m$
can be computed from $m$, $e$ and $k$ and hence
$C(\tau\,|\,m) \leq e+k+c$ for some constant $c$ independent of
$e,k,n,m$.

Now assume that $A$ is not infinitely often $O(1)\deepc$.
Then there is a total function $t = \varphi_e$
and a $k > c$ such that
$C(\tau\,|\,|\tau|) \geq C^t(\tau\,|\,|\tau|)-k$ for all prefixes 
$\tau$ of $A$. It follows that in particular never a $\sigma_m$ with
$S_{\sigma_m}$ consisting of prefixes of $A$ is selected
in the above algorithm using $e,k$. This then implies that
for almost all $n$ and the majority of the $m$ in the interval
from $4^n$ to $4^{n+1}$ (which are those for which $\sigma_m$ does
not get defined) it holds that $C^t(\tau\,|\,m) \leq e+3k$ for
the unique $\tau \in S_{A(1)A(4)A(16)\ldots A(4^n)} \cap \{0,1\}^m$.
There are at most $2^{e+3k+2}$ many strings $\sigma \in \{0,1\}^{n+1}$
such that at least half of the members $\tau$ of $S_{\sigma}$ satisfy
that $C(\tau\,|\,|\tau|) \leq e+3k$ and there is a constant $c'$
such that for almost all $n$ the corresponding $\sigma$ satisfy
$C(\sigma | n) \leq e+3k+c'$.
It follows that $C(\tau\,|\,|\tau|) \leq e+3k+c''$ for some constant
$c''$ and almost all $n$ and all $\tau \in S_{A(1)A(4)A(16)\ldots A(4^n)}$;
in other words, $C(\tau\,|\,|\tau|) \leq e+3k+c''$ for some constant
$c''$ and almost all prefixes $\tau$ of $A$. Hence $A$ is recursive
\cite[Exercise 2.3.4 on page 131]{Li2008}.

\medskip
\noindent
For the third item, note that Merkle, Kjos-Hanssen and Stephan
\cite[Theorem 2.7]{KMS11} showed that a set $A$ has DNR Turing
degree iff there is a function $f \leq_T A$ such that $C(f(n)) \geq n$
for all $n$. It will be shown that sets which are neither recursive
nor i.o.\ $O(1)\deepc$ will permit to construct such a function $f \leq_T A$
and are thus DNR.

Assume now that there is a time bound $t$ and a constant $c$ such that,
for all $n$, $C(A\harp n \ | \ n)+c \geq C^t(A\harp n \ | \ n)$.
Now, for input $n$, one searches relative to $A$ for an $m$ such
that $C^t(A \harp m \ | \ m) \geq n+c$ and lets
$f(n) = A \harp m$ for the so found $m$. As $A$ is not recursive this
search terminates for every $n$ and obviously $f \leq_T A$.
By assumption, $C(A \harp m \ | \ m)+c \geq n+c$ and, assuming
a suitable compatibility between conditional and normal Kolmogorov
complexity, $C(A \harp m) \geq n$, that is, $C(f(n)) \geq n$.
\end{proof}

\medskip
\noindent
Due to these connections, non-recursive
and non-high r.e.\ sets are a natural example
of sets where all members of the Turing degree satisfy that
they are i.o.\ $O(1)\deepc$ but not $O(1)\deepc$.

Results of Franklin and Stephan \cite{franklinstephan} imply that for every
Schnorr trivial set and every order $h$ it holds that $A$ is not
i.o.\ $h\deepc$, as for every order $h$ there is a time bound $t$
such that the function $n \mapsto C^t(A \harp n\,|\,n)$
grows slower than $h$. Thus, the second and third points cannot be
generalised to i.o.\ order$\deepc$.

It also follows that there are high truth-table degrees and hyperimmune-free
Turing degrees which do not contain any i.o.\ order$\deepc$ set. These are
obtained by considering examples for Schnorr trivial sets such as the
following ones: all maximal sets and, for every partial recursive
$\{0,1\}$-valued function $\psi$ whose domain is a maximal set,
all sets $A$ satisfying $\forall x\,\psi(x)\!\downarrow\,\Rightarrow
A(x) = \psi(x)$.

\section{Conclusion}

\noindent
We conclude that the choice of the depth magnitude has consequences on
the computational power of
the corresponding deep sets,
and that larger magnitudes is not necessarily preferable over smaller
magnitudes.
Therefore choosing the appropriate depth magnitude for one's purpose
is delicate,
as the corresponding depth notions might be very different. 
When the depth magnitude is large, we proved that depth and highness coincide.
We showed that this is not the case for smaller depth magnitude by
constructing a low order deep set, but the set is not r.e.
We therefore ask whether there is a low $O(1) \deepk$ r.e.\ set.

From our results, for magnitudes of order $O(1)$, $K$-depth behaves
better than $C$-depth.
To further strengthen that observation we ask whether there is  an MLR
$O(1) \deepc$ set.

\end{document}